\documentclass[11pt,a4paper]{article}

\usepackage{preamble}
\usepackage{authblk}

\title{The Dynamical Lie Algebra of QAOA-MaxCut on the Complete Graph}

\author{Jonathan Allcock  \thanks{jonallcock@tencent.com}}
\author{Pei Yuan  \thanks{peiyuan@tencent.com}}
\author{Shengyu Zhang \thanks{shengyzhang@tencent.com}}
\affil{Tencent Quantum Laboratory}
\date{\today}

\begin{document}
\maketitle

\begin{abstract}
    We give an analytical expression for the dynamical Lie algebra corresponding to the QAOA-MaxCut problem on complete graphs, and show that the variance of the associated loss function scales linearly in the number of qubits. This solves an open problem from \cite{allcock2026} and confirms that such systems do not exhibit barren plateaus. The proof is based on projecting the dynamical Lie algebra generators onto  subspaces given by the Schur-Weyl duality between irreducible representations of the unitary and symmetric groups.
\end{abstract}

\section{Introduction}\label{introduction}

Dynamical Lie algebras (DLAs) have proven to be a valuable tool in the study of variational quantum algorithms, providing a quantitative way of characterizing the variance of the associated loss functions~\cite{ragone2023unified, fontana2024characterizing}, and hence determining the presence or absence of barren plateaus~\cite{mcclean2018barren}. These are flat regions in parameter space which inhibit the efficient training of such algorithms, and which correspond to the exponential decay of the loss function variance as a function of the number of qubits in the system.

The DLAs corresponding to the Quantum Approximate Optimization Algorithm~\cite{farhi2014quantum} for the graph MaxCut problem (QAOA-MaxCut) have been studied in recent years, but have been characterized for only two specific graphs -- cycles~\cite{d2024controllability,allcock2026} and paths~\cite{kazi2024analyzing}. Recently, DLAs for random graphs~\cite{mao2025qaoa} and for the multi-angle (ma) variant of QAOA-MaxCut~\cite{herrman2022multi}, in which each gate has its own angle parameter, have also been characterized~\cite{kokcu2024classification,kazi2024analyzing}. These results show that QAOA on random graphs and ma-QAOA exhibit barren plateaus. This motivates further studies on QAOA for graphs with symmetries, for which the DLAs may have finer structure, preventing barren plateaus from occurring.

Complete graphs have the most symmetries of all graphs, but obtaining a characterization of their DLAs remains an open problem. For complete graphs $K_n$ on $n$ vertices, partial results for the QAOA-MaxCut DLA $\mathfrak{g}_{K_n}$ were obtained in~\cite{allcock2026}, where it was shown that $\mathfrak{g}_{K_n}\cong \mathfrak{c}\oplus [\mathfrak{g}_{K_n},\mathfrak{g}_{K_n}]$, where the abelian center $\mathfrak{c}$ has dimension one or two\footnote{$\mathfrak{c}\cong \mathfrak{u}(1)$ if $n$ is even, and $\mathfrak{c}\cong \mathfrak{u}(1)\oplus \mathfrak{u}(1)$ if $n$ is odd.}, the semisimple component $[\mathfrak{g}_{K_n},\mathfrak{g}_{K_n}]$ has dimension $\Theta(n^3)$, and the decomposition of $[\mathfrak{g}_{K_n},\mathfrak{g}_{K_n}]$ into simple ideals was conjectured to have the following form.  




\begin{conjecture} [\cite{allcock2026}] 
$[\mathfrak g_{K_n},\mathfrak g_{K_n}] \cong \bigoplus_{j=1}^{n-1}\mathfrak{su}\left( \lfloor (j+1)/2\rfloor + 1 \right)$.
\end{conjecture}

This work contains two contributions. We first prove the following theorem, which is equivalent to the above conjecture.

\begin{theorem} \label{thm:dla-decomp}
     $[\mathfrak g_{K_n},\mathfrak g_{K_n}]\cong \bigoplus_{j\in\mathcal{J}_n} \mathfrak{su}(\left\lfloor j\right\rfloor +1)\oplus \mathfrak{su}(\left\lceil j \right\rceil)$,  where $\mathcal{J}_n = \{0,1,\ldots, n/2\}$ if $n$ is even, and $\mathcal{J}_n=\{1/2, 3/2, \ldots, n/2\}$ if $n$ is odd, and with the convention that $\mfsu(0) = \mfsu(1) = 0$.
\end{theorem}

We then analyze the variance of the loss function. Recall that in a variational quantum algorithm, an initial state $\rho$ is passed into a parameterized quantum circuit $U(\theta)$, at the end of which an observable $O$ is measured.  The loss function is then taken to be 
\eq{
\ell(\rho,O; \theta) = \operatorname{tr}\left( U(\theta)\rho U(\theta)^{\dagger}O \right).
}

Our second result is that, for $K_n$, the loss function variance scales linearly with $n$.




\begin{theorem} \label{thm:loss-var}For QAOA-MaxCut on~$K_{n}$, for circuits sufficiently deep that $U(\theta)$ acts as a unitary $2$-design,
\[\operatorname{Var}_{\theta}[\ell(\rho,O;\theta)] = \begin{cases}
\frac{4n^{2} + n - 3}{45(n - 1)}, & {\rm if~ } n\text{\rm~is~even;} \\
\frac{4n + 8}{45}, & \text{\rm if~} n\text{\rm~is~odd}.
\end{cases}\]
\end{theorem}
Note that for both even and odd $n$,~$\operatorname{Var}_{\theta}[\ell(\rho,O; \theta)] \rightarrow \frac{4n}{45} = \Theta(n)$ as~$n \rightarrow \infty$. Since this does not exhibit exponential decay in $n$, barren plateaus are not present.

\section{Preliminaries}
For integer $n\ge 0$, let $S_n$ denote the symmetric group on $n$ elements, and let $\mfu(n)$ and $\mathfrak{su}(n)$ denote the real Lie algebras of skew-Hermitian and traceless skew-Hermitian $n\times n$ complex matrices, respectively, with the Lie bracket corresponding to the matrix commutator. For a finite-dimensional Hilbert space $V$, let $\mfu(V)$ and $\mathfrak{su}(V)$ denote the real Lie algebras of skew-Hermitian and traceless skew-Hermitian operators from $V$ to itself, respectively. It is clear that $\mfu(V)\cong\mfu(\dim(V))$ and $\mathfrak{su}(V)\cong\mathfrak{su}(\dim(V))$.  For a set of skew-Hermitian operators $\mathcal{A}=\{iA_1, \ldots, iA_k\}$ we denote by $\langle \mathcal{A} \rangle_{\text{Lie},\mathbb{R}}$ the dynamical Lie algebra generated by $\mathcal{A}$, i.e., the real vector space spanned by all nested commutators of the $iA_j$. For more background on Lie algebras and DLAs see e.g., \cite{allcock2026}.

For $\pi\in S_n$, let $P_\pi$ be the unitary representation of $\pi$ acting on $(\mathbb{C}^2)^{\otimes n}$, defined by its action on computational basis states $\ket{k_1}\otimes\ket{k_2}\otimes\ldots\otimes\ket{k_n}$:
\eq{
P_\pi(\ket{k_1}\otimes\ket{k_2}\otimes\ldots\otimes\ket{k_n}) &= \ket{k_{\pi^{-1}(1)}}\otimes \ket{k_{\pi^{-1}(2)}}\otimes \ldots \otimes \ket{k_{\pi^{-1}(n)}}.
}
We say that an operator $H$ is $S_n$-invariant if $[H,P_\pi]=0$ for all $\pi\in S_n$, and define
\eq{
\mfu(2^n)^{S_n}=\{ H \in \mfu(2^n): [H,P_\pi]=0, \quad \forall \pi\in S_n\}.
}
In what follows, capital letters e.g., $H, D, T, X$ will denote Hermitian operators, and $iH, iD, iT, iX$ etc.\ their skew-Hermitian counterparts.

Finally, for all integers $n\ge 0$, define the set of spins
\eq{
\mathcal{J}_n&=\begin{cases}
    \{n/2, n/2-1, \ldots, 0\}, & \text {if $n$ is  even};\\
    \{n/2, n/2-1, \ldots, 1/2\}, & \text {if $n$ is odd}.
\end{cases}
}

\section{Dynamical Lie algebra of $K_n$}

Let \[\X=\sum_{k=1}^n X_{k}, \qquad \Z=\sum_{k=1}^n Z_{k}, \qquad \text{and}\qquad \ZZ=\sum_{1\le s<t\le n}Z_{s} Z_{t},\] where $X_{k}$ and $Z_{k}$ are the Pauli $X$ and Pauli $Z$ operators acting on qubit $k$ of an $n$-qubit system, respectively. The DLA of QAOA-MaxCut on $K_n$ is defined to be 
\eq{
\mathfrak{g}_{K_n}=\langle \{i\X, i\ZZ\}\rangle_{\text{Lie},\mathbb{R}}.
}
As a subalgebra of $\mathfrak{u}(2^n)$, $\mathfrak{g}_{K_n}$ admits a decomposition 
\eql{ \label{eq:reductive}
\mathfrak{g}_{K_n}\cong \mathfrak{c}\oplus [\mathfrak g_{K_n},\mathfrak g_{K_n}],
}
where $\mathfrak{c}$ is an abelian center, and $[\mathfrak g_{K_n},\mathfrak g_{K_n}]=\spn_{\mathbb R}\{[A,B] : A,B\in\mathfrak{g}_{K_n} \}$ is a semisimple subalgebra. Note that \[\Z^2 = nI + 2\cdot \ZZ,\] and since the identity $I$ on $n$ qubits is annihilated by the commutator, we have the following fact.
\begin{fact} Let $\mathfrak{g}=\langle \{i\X, i\Z^2\}\rangle_{{\rm Lie},\mathbb{R}}$. Then \[[\mfg,\mfg]=[\mathfrak{g}_{K_n},\mathfrak{g}_{K_n}].\]
\end{fact}

In this section we prove the following theorem which, by the above fact, implies Theorem~\ref{thm:dla-decomp}.


\begin{theorem}\label{thm:gg} Let $\mfg =\langle \{i\X, i\Z^2\}\rangle_{{\rm Lie}, \mathbb{R}}$. Then \[[\mfg,\mfg] \cong \bigoplus_{j\in\mathcal{J}_n} \mathfrak{su}(\left\lfloor j\right\rfloor +1)\oplus \mathfrak{su}(\left\lceil j \right\rceil).\]
\end{theorem}

Our proof of this theorem is based on the fact that $\X$ and $\Z$ are $S_n$-invariant Hermitian operators, and thus Schur-Weyl duality can be applied to analyzing the decomposition of the space of $n$ qubits into tensor products of irreducible representations of  $\mathfrak{su}(2)$ and $S_n$. 
\begin{theorem}[Schur-Weyl duality~\cite{fulton2013representation}] \label{thm:schur-weyl}
\eq{
(\mathbb{C}^2)^{\otimes n}\cong \bigoplus_{j\in \mathcal{J}_n}W_j \otimes V_j,
}
where $W_j$ and $V_j$ are irreducible representations of $\mathfrak{su}(2)$ and $S_n$, respectively.
\end{theorem}
The irreducible representations $W_j$ and $V_j$ appearing in Theorem~\ref{thm:schur-weyl} are fully classified, and have dimensions $\dim(W_j)=2j+1$ and $\dim(V_j)={n \choose n/2-j}- {n \choose n/2 -j -1}$.

Schur-Weyl duality implies that there exists a basis transformation such that all $iH\in \mfu(2^n)^{S_n}$ simultaneously block diagonalize as $iH=i\bigoplus_j H_j \otimes I_{V_j}$, where $iH_j\in\mathfrak{u}(W_j)$.  Since $\X,\Z\in \mfu(2^n)^{S_n}$ and, furthermore, since $\X$ and $\Z$ generate a representation of $\mathfrak{su}(2)$ acting on $(\mathbb{C}^2)^{\otimes n}$, 
the $iH_j$ are irreducible representations of $\mathfrak{su}(2)$. By  the standard Clebsch-Gordan construction~\cite{sakurai2020modern}, one can express these as
\eql{
i\X &= i\bigoplus_{j\in\mathcal{J}_n} 2J_x^{(j)}\otimes I_{V_j}, \qquad i\Z = i\bigoplus_{j\in\mathcal{J}_n} 2J_z^{(j)}\otimes I_{V_j}, \label{eq:XZ}
}
where, for each $j\in \mathcal{J}_n$,  $J_x^{(j)}$ and $J_z^{(j)}$ can be defined with respect to a basis $\{\ket{j,m}_x\}_{m\in \mathcal{M}_j}$ by
\eql{
J_x^{(j)}\ket{j,m}_x &= m\ket{j,m}_x, \label{eq:jx}\\
J_z^{(j)}\ket{j,m}_x &= -\frac{i}{2}\left(\sqrt{j(j+1)-m(m+1)}\ket{j,m+1}_x - \sqrt{j(j+1)-m(m-1)}\ket{j,m-1}_x\right), \label{eq:jz}
}
where the index set
\eq{
\mathcal{M}_j = \{j, j-1, \ldots, -j\}.
}

The basis states $\ket{j,m}_x$ can be explicitly constructed from tensor products of the single-qubit computational $X$-basis states, $\ket{+}$ and $\ket{-}$, via the application of Young symmetrizers \cite{fulton2013representation}. In particular, the maximal spin sector $j=n/2$ corresponds to the fully symmetric subspace, where the highest-weight state is simply the product state 
\eql{ \label{eq:maxjstate}
\ket{n/2, n/2}_x = \ket{+}^{\otimes n}.
}
Indeed, $J_x^{(n/2)} \ket{+}^{\otimes n} = \frac{\X}{2} \ket{+}^{\otimes n} = \frac{n}{2} \ket{+}^{\otimes n}$, therefore $\ket{+}^{\otimes n}$ is the eigenvector of $J_x^{(n/2)}$ with eigenvalue $m=n/2$.

From Eq.~\eqref{eq:XZ}, it follows that \[\langle \{i\X, i\Z^2\} \rangle_{\text{Lie},\mathbb{R}} \cong\langle \{i\bigoplus_j J_x^{(j)}\otimes I_{V_j}, i\bigoplus_j(J_z^{(j)})^2\otimes I_{V_j}\}\rangle_{\text{Lie},\mathbb{R}}\] and, as the next lemma shows, we can ignore the multiplicity (i.e., the $I_{V_j}$ term) of each $J_{x,z}^{(j)}$ block (when the Lie algebra structure is concerned).  Define \[J_x=\bigoplus_j J_x^{(j)} \qquad \text{and} \qquad J_z^2 = \bigoplus_j\left(J_z^{(j)}\right)^2.\]  

\begin{lemma}\label{lem:la-iso}$\langle \{i\X, i\Z^2 \}\rangle_{{\rm Lie},\mathbb{R}} \cong \langle \{iJ_x, iJ_z^2 \}\rangle_{{\rm Lie},\mathbb{R}}$. 
\end{lemma}

\begin{proof} Let $\mathbb{C}I_{V_j} = \{c I_{V_j} : c \in \mathbb{C}\}$ be the set of linear operators proportional to the identity acting on the space $V_j$, and define the map $\phi:\bigoplus_j\left(\mfu(W^{(j)})\otimes \mathbb{C}I_{V_j}\right)\rightarrow \bigoplus_j\mfu(W^{(j)})$ by $\phi\left(i \bigoplus_j  H_j\otimes I_{V_j}\right)= i \bigoplus_j H_j$. 
It is easy to verify that 
\begin{enumerate}
    \item $\phi\left[ i\bigoplus_j H_j \otimes I_{V_j}, i\bigoplus_j H'_j\otimes I_{V_j}\right] = \left[ \phi\left(i\bigoplus_j H_j\otimes I_{V_j}\right), \phi\left(i\bigoplus_j H'_j\otimes I_{V_j}\right)\right]$,
    \item $\phi\left(i\bigoplus_j H_j \otimes I_{V_j}\right)=0 \Rightarrow i\bigoplus_j H_j\otimes I_{V_j}=0$, 
\end{enumerate}  
and thus $\phi$ is a Lie algebra homomorphism with a trivial kernel, i.e., a Lie algebra isomorphism. The lemma follows from applying $\pi$ to $i\X$ and $i\Z^2$ from Eq.~\eqref{eq:XZ}. 
\end{proof}

Lemma~\ref{lem:la-iso} allows us to identify $\mfg \cong \left\langle \{iJ_x, iJ_z^2\}\right\rangle_{\text{Lie},\mathbb{R}}$, and motivates us to analyze the operators $J_x^{(j)}$ and $\left(J_z^{(j)}\right)^2$.  We will find the following facts, which can be verified by direct computation from their definitions, useful.

\begin{fact}\label{fact:jx-jz1}
    The Hermitian operator $\left(J_z^{(j)}\right)^2$ decomposes into a diagonal and bi-diagonal part $\left(J_z^{(j)}\right)^2 =  D^{(j)} + T^{(j)}$, where 
    \eq{D^{(j)}&= \frac{1}{2}\sum_{m\in \mathcal{M}_j}\left(j(j+1)-m^2\right)\ket{j,m}_x\bra{j,m}_x,\\
    T^{(j)}&= \sum_{m=-j}^{j-2}t^{(j)}_{m}\big(\ket{j,m+2}_x\bra{j,m}_x + \ket{j,m}_x\bra{j,m+2}_x\big),
    }
    with
    \eql{
    t^{(j)}_{m}=-\frac{1}{4}\sqrt{j(j+1)-m(m+ 1)}\sqrt{j(j+1)-(m+ 1)(m+ 2)}. \label{eq:tjm}}
\end{fact}
It will be important to note that $t_m^{(j)}$  is non-zero for all $j\in \mathcal{J}_n$ and all $m\in \{j-2,\ldots, -j\}$.




\begin{definition} \label{def:Wspaces} 
Define the following spaces
\begin{enumerate}
    \item $W^{(j)} = \spn\{ \ket{m,j}_x : m\in \mathcal{M}_j\}$
    \item $W^{(j)}_{\text{even}} =\spn\{\ket{j,m}_x : (j-m) \in 2\mathbb{Z}\}$
    \item $W_\text{odd}^{(j)}=\spn\{\ket{j,m}_x : (j-m) \in 2\mathbb{Z}+1\}$
\end{enumerate} 
and define
\begin{align*}
    \mathcal{M}_{j,\text{even}} & = \{m\in \mathcal{M}_j : (j-m)\in 2\mathbb{Z}\},\\
    \mathcal{M}_{j,\text{odd}} & = \{m\in \mathcal{M}_j : (j-m)\in 2\mathbb{Z}+1\}.
\end{align*}
\end{definition}
\noindent It follows from these definitions that, for any $j\in \mathcal{J}_n$, 
\eql{
\dim\left(W^{(j)}_\text{even}\right)&= \lfloor j\rfloor + 1, \qquad \dim\left(W^{(j)}_\text{odd}\right)= \lceil j \rceil. \label{eq:dimWevenodd}
}

\begin{lemma} $W^{(j)}_{{\rm even}}$ and $W_{\rm odd}^{(j)}$ are both invariant subspaces of $J_x^{(j)}$ and $\left(J^{(j)}_z\right)^2$.
\end{lemma}
\begin{proof}
    Follows directly from Eq.~\eqref{eq:jx} and Fact~\ref{fact:jx-jz1}.
\end{proof}
The invariance of these two subspaces for all $j\in \mathcal{J}_n$ immediately implies the following. 

 \begin{lemma}\label{lem:g-gg} \par\noindent 
 \begin{enumerate}
     \item $\mathfrak{g}\subseteq \bigoplus_j\mathfrak{u}(W^{(j)}_{\rm even}) \oplus \mathfrak{u}(W^{(j)}_{\rm odd})$.
     \item $[\mfg,\mfg]\subseteq \bigoplus_j\mathfrak{su}(W^{(j)}_{\rm even})\oplus \mathfrak{su}(W^{(j)}_{\rm odd})\cong  \bigoplus_{j\in\mathcal{J}_n} \mathfrak{su}(\left\lfloor j\right\rfloor +1)\oplus \mathfrak{su}(\left\lceil j \right\rceil)$. \label{line:g-gg-2}
 \end{enumerate}
 \end{lemma}

From Eq.~\eqref{eq:dimWevenodd} it follows that
\eql{ 
\dim([\mfg,\mfg]) &\le \sum_{j\in\mathcal{J}_n} \big(\dim\left(\mathfrak{su}(\left\lfloor j\right\rfloor +1)\right) + \dim\left(\mathfrak{su}(\left\lceil j \right\rceil)\right)\big) \notag\\
&=\begin{cases}
       \frac{1}{12} (n^3+6n^2+2n ), & \quad \text{if $n$ is even;}\\
        \frac{1}{12} (n^3+6n^2-n-6), &\quad \text{if $n$ is odd.}
\end{cases}\label{eq:dimgg}
}
where the second line uses the fact that $\dim(\mathfrak{su}(N)) = N^2-1$ for $N > 1$ and is zero for $N=0,1$. However, the dimension of $[\mfg,\mfg]$ is known (\cite{allcock2026}) to be exactly equal to Eq.~\eqref{eq:dimgg}, and thus the inclusion $\subseteq$ in Lemma~\ref{lem:g-gg} \eqref{line:g-gg-2} must actually be equality, i.e, 
\eq{
[\mfg,\mfg]=\bigoplus_j\mathfrak{su}(W^{(j)}_\text{even})\oplus \mathfrak{su}(W^{(j)}_\text{odd})\cong  \bigoplus_{j\in\mathcal{J}_n} \mathfrak{su}(\left\lfloor j\right\rfloor +1)\oplus \mathfrak{su}(\left\lceil j \right\rceil).
}
This proves Theorem~\ref{thm:gg}.

However, this proof has two unsatisfactory aspects. First, it relies on the dimension bound from~\cite{allcock2026}, which was based on a lengthy and complicated analysis of a basis for the DLA. Second, the matching dimensions force the equality of spaces, but it does not reveal \textit{how} the entire space $\mathfrak{su}(W^{(j)}_\text{even})\oplus \mathfrak{su}(W^{(j)}_\text{odd})$ is generated from the generators $iJ_x^{(j)}$ and $i\big(J_z^{(j)}\big)^2$. 

In the remainder of this section we give a concise ab initio proof of the following reverse inclusion which, along with Eq.~\eqref{eq:dimWevenodd}, completes the proof of Theorem~\ref{thm:gg} in a self-contained fashion.

 \begin{lemma}\label{lem:reverse-inclusion} $\bigoplus_j\mathfrak{su}(W^{(j)}_{\rm even})\oplus \mathfrak{su}(W^{(j)}_{\rm odd})\subseteq[\mfg,\mfg]$.
 \end{lemma}

To prove Lemma~\ref{lem:reverse-inclusion}, we will show that a basis for $\mathfrak{su}\left(W^{(j)}_\text{even}\right)$ and $\mathfrak{su}\left(W^{(j)}_\text{odd}\right)$ is in $[\mfg, \mfg]$, for all $j$.   We first note the following fact, which can be verified by direct computation.
\begin{fact}\label{fact:jx-jz2}
Applying the adjoint maps ${\rm ad}_{iJ_x^{(j)}}$ and ${\rm ad}_{iT^{(j)}}$ twice, to $i(J_z^{(j)})^2$ and $iJ_x^{(j)}$, respectively, has the following effect.
\begin{enumerate}
    \item $[iJ_x^{(j)},[iJ_x^{(j)}, i(J_z^{(j)})^2]]=-4iT^{(j)}$.
    \item $[iT^{(j)},[iT^{(j)},iJ_x^{(j)}]] = -i\left[2j(j+1)-1\right]J_x^{(j)} + 2i\left(J_x^{(j)}\right)^3$.
\end{enumerate}
\end{fact}

\begin{lemma}\label{lem:things-in-g} 
    Define $T=\bigoplus_j T^{(j)}$ and $D=\bigoplus_j D^{(j)}$. The following are in $\mfg$: (1) $iT$, (2) $iD$, and (3) $-i\bigoplus_j \left[2j(j+1)-1\right]J_x^{(j)} + 2iJ_x^3$. 
\end{lemma}

\begin{proof}
    Applying Fact~\ref{fact:jx-jz2} to $[iJ_x,[iJ_x,i(J_z)^2]]$ (for all $j$) gives $iT\in \mfg$. 
    It then follows that $iD\in\mfg$ since $i(J_z)^2=iD + iT$ is one of the generators of $\mfg$. The third fact follows from Fact~\ref{fact:jx-jz2}  and the definition of $J_x = \bigoplus_j J_x^{(j)}$.
\end{proof}

Now, for a vector space $V$ with orthonormal basis $\{\ket{1},\ldots, \ket{d}\}$, a standard basis for $\mathfrak{su}(V)$ consists of $\{iX_{p,q}, iY_{p,q}\}_{p,q\in [d], q < p}\cup \{iZ_{p+1,p}\}_{p\in [d-1]}$, where 
\begin{align*}
    iX_{p,q}&= i\left(\ket{p}\bra{q} + \ket{q}\bra{p}\right),\\
    iY_{p,q}&= \ket{p}\bra{q} - \ket{q}\bra{p},\\
iZ_{p+1,p} &= i\left(\ket{p+1}\bra{p+1} - \ket{p}\bra{p}\right).
\end{align*} 
Note that
\eq{
    [iX_{p+1,p}, iY_{p+1,p}] =-2i Z_{p+1,p}, \qquad [iX_{p,q}, iX_{q,r}] = -iY_{p,r}, \qquad 
[iX_{p,q}, iY_{q,r}] = iX_{p,r}, 
}
and therefore all basis vectors can be generated by taking successive commutators of  $iX_{p+1,p}$, $ iY_{p+1,p}$.  Since $W^{(j)}_\text{even}$ and $W^{(j)}_\text{odd}$ are spanned by basis vectors $\ket{j,m}_x$ whose $m$ indices differ by multiples of two (see Def.~\ref{def:Wspaces}),  the following lemma suffices to show that $[\mfg,\mfg]$ contains a basis for all $\mathfrak{su}\left(W^{(j)}_\text{even}\right)$ and $\mathfrak{su}\left(W^{(j)}_\text{odd}\right)$.

\begin{lemma}\label{lem:XYing} For all $j\in \mathcal{J}_n$ and all $m\in\{j-2, \ldots, -j\}$ the following are in $[\mfg,\mfg]$.
\eq{
iX^{(j)}_{m+2,m}&=i\left(\ket{j,m+2}_x\bra{j,m}_x + \ket{j,m}_x\bra{j,m+2}_x\right),\\
iY^{(j)}_{m+2,m}&=\ket{j,m+2}_x\bra{j,m}_x - \ket{j,m}_x\bra{j,m+2}_x.
}
\end{lemma}

\begin{proof}
From Lemma~\ref{lem:things-in-g} we have that $iT,~iD,~i\bigoplus_j\left[2j(j+1)-1\right]J_x^{(j)}-2iJ_x^3 \in \mfg$, and we shall show how to use these to generate $iX^{(j)}_{m+2,m}$ and $iY^{(j)}_{m+2,m}$.  Note that any operator $A^{(j)} = \sum_{m\in \mathcal{M}_j}a_m\ket{j,m}_x\bra{j,m}_x$ that is diagonal in the $\{\ket{j,m}_x\}$ basis satisfies
\eql{
[iA^{(j)},iX_{m+2,m}^{(j)}] &= -(a_{m+2}-a_m)(iY^{(j)}_{m+2,m}),\label{eq:commAX}\\
[iA^{(j)},iY_{m+2,m}^{(j)}] &= (a_{m+2}-a_m)(iX^{(j)}_{m+2,m}).\label{eq:commAY}
}

Using the definitions of $D^{(j)}$ and $T^{(j)}$ from Fact~\ref{fact:jx-jz1}, and noting that $D^{(j)}$ and  $J_x^{(j)}$ are diagonal in the $\{\ket{j,m}_x\}$ basis, and that $T^{(j)} =\sum_{m=-j}^{j-2} t_m^{(j)} X_{m+2,m}^{(j)}$, the following are easy to verify:
\eql{
   [iD^{(j)}, iT^{(j)}]&= \sum_{m=-j}^{j-2} 2(m+1) t^{(j)}_m (iY^{(j)}_{m+2,m}),\label{eq:DT}\\  
     [iJ_x^{(j)}, iT^{(j)}]&=-2\sum_{m=-j}^{j-2}  t^{(j)}_m(iY^{(j)}_{m+2,m}),\\ 
\left[i\left(J_x^{(j)}\right)^3, iT^{(j)}\right]&=-2\sum_{m=-j}^{j-2}  (3m^2+6m+4)t^{(j)}_m(iY^{(j)}_{m+2,m}). \label{eq:J3T}
}
Now, for $\alpha,\beta\in \mathbb{R}$, define 
\eq{
iA_{\alpha,\beta} &= \left(i\bigoplus_j\left[2j(j+1)-1\right]J_x^{(j)} - 2iJ_x^3 \right) +(3-\beta) iJ_x+\alpha iD
}
which is in $\mfg$.  Write $A_{\alpha,\beta} = \bigoplus_j A_{\alpha,\beta}^{(j)}$. Using Eqs.~\eqref{eq:DT}-\eqref{eq:J3T}, it can be verified that, 
\eq{
{\rm ad}_{iA^{(j)}_{\alpha,\beta}}(iT^{(j)}) = \lambda_m^{(j)}(\alpha,\beta)(iY_{m+2,m}^{(j)}),
}
where
\eql{
\lambda^{(j)}_m(\alpha,\beta)=-4j(j+1) + 12(m+1)^2 + 2\beta +2\alpha(m+1). \label{eq:lambda-mj}
}
Furthermore, using Eqs.~\eqref{eq:commAX} and \eqref{eq:commAY}, it follows that, for any integer $r\ge 1$, 

\eql{
{\rm ad}_{iA_{\alpha,\beta}}^{2r}(iT) &= \sum_{j\in \mathcal{J}_n}\sum_{m=-j}^{j-2} \left(-\left(\lambda^{(j)}_m(\alpha,\beta)\right)^2\right)^{r}t^{(j)}_m (iX^{(j)}_{m+2,m}), \label{eq:adA}\\ 
{\rm ad}_{iA_{\alpha,\beta}}^{2r+1}(iT) &= (-1)^r\sum_{j\in \mathcal{J}_n}\sum_{m=-j}^{j-2} \left(\lambda^{(j)}_{m}(\alpha,\beta)\right)^{2r +1}t^{(j)}_m (iY^{(j)}_{m+2,m}).\label{eq:adA2} 
}

To simplify notation, let $\mathcal{K}=\{(j,m) : j\in \mathcal{J}_n, m\in\{j-2, \ldots, -j\}\}$ be the set of all valid $(j,m)$ pairs, and $K=\abs{\mathcal{K}}$.  For $k=(j,m)\in \mathcal{K}$, we will denote 
\eq{
\lambda_k(\alpha,\beta) &= \lambda^{(j)}_m (\alpha,\beta), \qquad \sigma_k(\alpha,\beta) = -\left(\lambda_m^{(j)}(\alpha,\beta)\right)^2, \qquad t_k = t_m^{(j)},\\
i\tilde{X}_k &= iX^{(j)}_{m+2,m}, \qquad i\tilde{Y}_k = iY^{(j)}_{m+2,m}.
}

Then, Eq.~\eqref{eq:adA} and Eq.~\eqref{eq:adA2} can be expressed (suppressing the $\alpha,\beta$ parameters for notational simplicity) as the Vandermonde matrix equations

\eq{
\begin{pmatrix} 
\text{ad}_{iA_{\alpha,\beta}}^0(iT) \\ 
\text{ad}_{iA_{\alpha,\beta}}^2(iT) \\ 
\vdots \\ 
\text{ad}_{iA_{\alpha,\beta}}^{2K-2}(iT) 
\end{pmatrix} 
= 
\begin{pmatrix} 
1 & 1 & \dots & 1 \\ 
\sigma_1 & \sigma_2 & \dots & \sigma_K \\ 
\vdots & \vdots & \ddots & \vdots \\ 
\sigma_1^{K-1}& \sigma_2^{K-1} & \dots & \sigma_K^{K-1} 
\end{pmatrix} 
\begin{pmatrix} 
t_1 \left(i\tilde{X}_1\right) \\ 
t_2 \left(i\tilde{X}_2\right) \\ 
\vdots \\ 
t_K \left(i\tilde{X}_K\right)
\end{pmatrix}
}
and
\eq{
\begin{pmatrix} 
\text{ad}_{iA_{\alpha,\beta}}^1(iT) \\ 
\text{ad}_{iA_{\alpha,\beta}}^3(iT) \\ 
\vdots \\ 
\text{ad}_{iA_{\alpha,\beta}}^{2K-1}(iT) 
\end{pmatrix} 
= 
\begin{pmatrix} 
1 & 1 & \dots & 1 \\ 
\sigma_1 & \sigma_2 & \dots & \sigma_K \\ 
\vdots & \vdots & \ddots & \vdots \\ 
\sigma_1^{K-1}& \sigma_2^{K-1} & \dots & \sigma_K^{K-1} 
\end{pmatrix} 
\begin{pmatrix} 
t_1 \lambda_1 \left(i\tilde{Y}_1\right) \\ 
t_2 \lambda_2 \left(i\tilde{Y}_2\right) \\ 
\vdots \\ 
t_K\lambda_K \left(i\tilde{Y}_K\right)
\end{pmatrix},
}
respectively.  In Appendix~\ref{app:app1}, we show that for $\alpha, \beta\in\mathbb{R}$ such that $\{1,\alpha,\beta\}$ are linearly independent over $\mathbb{Q}$, $\sigma_k(\alpha,\beta) =\sigma_{k'}(\alpha,\beta)\Rightarrow k = k'$ and $\lambda_k(\alpha,\beta)\neq 0$. Since $t_k\neq 0$ for all $k$, it follows that, for such a choice of $\alpha,\beta$, both Vandermonde matrices can be inverted, showing that $i\tilde{X}_k=iX^{(j)}_{m+2,m}$ and $i\tilde{Y}_k=iY^{(j)}_{m+2,m}$ can be expressed as linear combinations of $\text{ad}_{iA_{\alpha,\beta}}^r(iT)$.
\end{proof}

\section{Loss function variance}\label{sec:variance}

In this section we prove Theorem~\ref{thm:loss-var}. To do so, we will make use of the following result, which relates the loss function variance to the DLA.
\begin{theorem}[\cite{ragone2023unified}]\label{thm:ragone}
    For a variational quantum algorithm with DLA $\mathfrak{g}$ satisfying $[\mfg,\mfg]=\bigoplus_j \mathfrak{g}_j$, if $O\in i\.g$ or $\rho\in i\.g$, and the circuit is deep enough to form a unitary 2-design, then 
\begin{equation}\label{eq:VQA-exp-var}
    \operatorname{Var}_{\theta}[\ell(\rho,O;\theta)] = \sum_{j=1}^k \frac{\+P_{\.g_j}(\rho) \+P_{\.g_j}(O)}{\dim(\.g_j)}.
\end{equation}
Here $\+P_{\.s}(H)$, the $\.s$-purity of a Hermitian operator $H\in i\.{u}(2^n)$, with respect to a subalgebra $\.s$ of $\.u(2^n)$, is defined as $\+P_{\.s}(H)= \tr(H_{\.s}^2)$
    where $H_\.s$ is the orthogonal projection of $H$ onto $\.s_{\mbC}$ (the complexification of $\.s$).
\end{theorem}

We wish to apply the above theorem to \[[\mathfrak{g}_{K_n},\mathfrak{g}_{K_n}]=[\mfg,\mfg]=\bigoplus_j\mathfrak{su}(W^{(j)}_\text{even})\oplus \mathfrak{su}(W^{(j)}_\text{odd}).\] Recall that for QAOA-MaxCut on $K_n$, the initial state $\rho=\ket{+}^{\otimes n}\bra{+}^{\otimes n}$, and the observable $O=\frac{1}{\sqrt{|E|}}\ZZ$, where $|E|=n(n-1)/2$ is the number of edges in $K_n$.

In what follows, denote by $\Pi^{(n/2)}_\text{even}$ the projector onto the $W^{(n/2)}_\text{even}$ subspace, and \[d=\dim\left(W^{(n/2)}_{\text{even}}\right) = \left\lfloor n/2\right\rfloor +1.\] Let $\mathfrak{s}=\mathfrak{su}\left(W^{(n/2)}_{\text{even}}\right)$, and note that
\eql{ \label{eq:dim-s}
\dim(\mathfrak{s}) &= \left(\lfloor n/2 \rfloor + 1\right)^2 - 1 = 
\begin{cases} 
\frac{n(n+4)}{4}, & \text{if } n \text{ is even}; \\
\frac{(n+3)(n-1)}{4}, & \text{if } n \text{ is odd}.
\end{cases}
}

Now, Eq.~\eqref{eq:maxjstate} says that $\ket{+}^{\otimes n}=\ket{n/2,n/2}_x\in W^{(n/2)}_\text{even}$ and thus $i\rho\in\mathfrak{u}\left(W^{(n/2)}_{\text{even}}\right)$. Since the initial state lies in the $W^{(n/2)}_\text{even}$ subspace, only the simple subalgebra $\mathfrak{s}=\mathfrak{su}\left(W^{(n/2)}_{\text{even}}\right)$ has non-zero purity $\+P_\mathfrak{s}(\rho)$.  We therefore need only consider the single term $\frac{\+P_\mathfrak{s}(\rho)\+P_\mathfrak{s}(O)}{\dim(\mathfrak{s})}$ in Theorem~\ref{thm:ragone}.   Theorem~\ref{thm:loss-var} then follows by combining Eq.~\eqref{eq:dim-s} with the following lemma.

\begin{lemma} \label{lem:rho-o-purity} The initial state $\rho$ and observable $O$ have $\mathfrak{s}$-purities as follows:
    \begin{enumerate}
        \item  $\mathcal{P}_\mathfrak{s}(\rho) = \begin{cases}
        \frac{n}{n+2},&\quad {\rm if~} n {\rm~is~even};\\ 
        \frac{n-1}{n+1},&\quad {\rm if~} n \text{\rm~is~odd}.
        \end{cases}$\label{line:point1}
        \item $\mathcal{P}_{\mathfrak{s}}(O) = \begin{cases}
\frac{(n + 1)(n + 2)(n+4)(4n - 3)}{180(n - 1)}, &\quad  {\rm if~} n {\rm~is~even}; \\
\frac{(n + 1)(n + 2)(n + 3)}{45},\ \ \ \ \  &\quad  {\rm if~} n {\rm~is~odd}.
\end{cases}$\label{line:point2}
\end{enumerate}
\end{lemma}
\begin{proof} 
For Point \ref{line:point1}, note that $\rho = \frac{\Pi^{(n/2)}_\text{even}}{d} + \rho'$ where $\rho' \in \mathfrak{su}\left( W^{(n/2)}_{\text{even}} \right)$ is traceless. Then
\eq{
\mathcal{P}_{\mathfrak{s}}(\rho) & = \operatorname{tr}\left( \left( \rho' \right)^{2} \right) = \operatorname{tr}\left( \left( \rho - \frac{\Pi^{(n/2)}_\text{even}}{d} \right)^{2} \right) \\
      & = \operatorname{tr}\left( \rho^{2} \right) - 2\operatorname{tr}\left(\rho\Pi^{(n/2)}_\text{even}/d \right) + \operatorname{tr}\left( {\Pi^{(n/2)}_\text{even}/d^{2}} \right) \\
      &=1 - \frac{1}{d}.
}
Substituting $d=\lfloor n/2 \rfloor +1$   gives the result.

    

For Point \ref{line:point2}, first note that 
\eq{
O &=\frac{1}{\sqrt{|E|}}\ZZ= \frac{1}{2\sqrt{|E|}}\left( \Z^{2} - nI \right) = \frac{1}{\sqrt{|E|}}\left( 2 \bigoplus_{j\in \mathcal{J}_n}\left(J_{z}^{(j)}\right)^2 \otimes I_{V_j} - \frac{n}{2}I \right).
}
Thus the projection to the $j = n/2$  block gives 
\eq{
O^{(n/2)} &= \frac{1}{\sqrt{|E|}}\left( 2\left(J^{(n/2)}_z\right)^2 \otimes I_{V_{n/2}} - \frac{n}{2}I_{W_{n/2}} \otimes I_{V_{n/2}} \right).
}
As $V_{n/2}$ is 1-dimensional, we drop the $I_{V_{n/2}}$ term and write \[O^{(n/2)} = \frac{1}{\sqrt{|E|}}\left( 2\left(J^{(n/2)}_z\right)^2  - \frac{n}{2}I_{W_{n/2}} \right).\] Now let us further project it onto the $W^{(n/2)}_\text{even}$ space to give 
\eq{
O^{(n/2)}_\text{even} = \Pi^{(n/2)}_\text{even} O^{(n/2)}\Pi^{(n/2)}_\text{even},
}
and take the traceless part to obtain 
\eq{
O_{\mathfrak{s}} &= O^{(n/2)}_\text{even} - \frac{\operatorname{tr}\left(O^{(n/2)}_\text{even} \right)}{d}I_{W^{(n/2)}_\text{even}}.
}
The purity is

\eq{
\mathcal{P}_{\mathfrak{s}}(O) &= \operatorname{tr}\left( \left(O^{(n/2)}_\text{even} - \frac{\operatorname{tr}\left(O^{(n/2)}_\text{even} \right)}{d}I_{W^{(n/2)}_\text{even}} \right)^2\right) \\
&= \operatorname{tr}\left( \left(O^{(n/2)}_\text{even}\right)^2 \right) - \frac{\operatorname{tr}\left(O^{(n/2)}_\text{even}\right)^2}{d}
}

Finally, note that the eigenvalues of $O^{(n/2)}_\text{even}$ are $\left( m^{2} - n/2 \right)/\sqrt{|E|}$, where $m\in\mathcal{J}_n$ (see Appendix~\ref{app:app2}).  Noting that $|E| = n(n-1)/2$ we obtain,
\eq{
\mathcal{P}_{\mathfrak{s}}(O) &= \frac{2}{n(n-1)}\left(\sum_{j\in \mathcal{J}_n} (2j^2-n/2)^2 - \frac{1}{\lfloor n/2\rfloor +1}\left(\sum_{j\in\mathcal{J}_n}(2j-n/2)\right)^2\right)\\
&= \begin{cases}
\frac{(n + 1)(n + 2)(n+4)(4n - 3)}{180(n - 1)}, &\quad  \text{if } n \text{ is even}; \\
\frac{(n + 1)(n + 2)(n + 3)}{45},\ \ \ \ \  &\quad    \text{if } n \text{ is odd}.
\end{cases}
}
\end{proof}

\bibliographystyle{alpha}
\bibliography{dla.bib}

\appendix

\section{Distinctness of eigenvalues in Lemma~\ref{lem:XYing}} \label{app:app1}
Here we prove the following result, used in the proof of Lemma~\ref{lem:XYing}.

\begin{lemma}
For $k=(j,m)$ and $k'=(j',m')$ and $\alpha,\beta\in\mathbb{R}$ such that $\{1,\alpha,\beta\}$ are linearly independent over $\mathbb{Q}$, $\lambda_k(\alpha,\beta)\neq 0$ and  $\sigma_k(\alpha,\beta) = \sigma_{k'}(\alpha,\beta)\Rightarrow k=k'$.
\end{lemma}

\begin{proof}
Recall that $\sigma_k(\alpha,\beta) = -\left(\lambda^{(j)}_m(\alpha,\beta)\right)^2$, where
\eql{
\lambda^{(j)}_m(\alpha,\beta)=-4j(j+1) + 12(m+1)^2 + 2\beta +2(m+1)\alpha. \label{eq:lambda_ab}
}

Since $\{1,\alpha,\beta\}$ are linearly independent over $\mathbb{Q}$ and the coefficient of $\beta$ is not zero in Eq.~\eqref{eq:lambda_ab}, it follows that  $\lambda_m^{(j)}(\alpha,\beta)\neq 0$.  Now, if $\sigma_k(\alpha,\beta)=\sigma_{k'}(\alpha,\beta)$ then $\lambda_m^{(j)}(\alpha,\beta) = \pm \lambda_{m'}^{(j')}(\alpha,\beta)$.

\paragraph{Case 1: $\lambda^{(j)}_m(\alpha,\beta) = \lambda^{(j')}_{m'}(\alpha,\beta)$.} In this case we must have $m=m'$ for the two $\alpha$ terms to be equal.  But then equality of the rational terms requires that $j(j+1) = j'(j'+1)$. And since $j,j'\ge 0$, this implies that $j=j'$.

\paragraph{Case 2: $\lambda^{(j)}_m(\alpha,\beta) = -\lambda^{(j')}_{m'}(\alpha,\beta)$.} In this case we have
\eq{
-4j(j+1) - 4j'(j'+1) + 12 (m+1)^2 + 12(m'+1)^2 + 4\beta + 2(m+m'+2)\alpha &=0,
}
but this has no solutions since the coefficient of $\beta$ is not zero.
\end{proof}

\section{Eigenvalues of $\left( J_z^{(j)}\right)^2$ on $W^{(j)}_{\rm even}$}\label{app:app2}
Here we prove the following result which is used in the proof of Lemma~\ref{lem:rho-o-purity}.
\begin{lemma} The eigenvalues of $\left( J_z^{(j)}\right)^2$ restricted to the $W^{(j)}_{\rm even}$ subspace are $m^2$, for all $m\in\mathcal{M}_j^+$, each with multiplicity 1, where
\eq{
\mathcal{M}_j^+ &= \begin{cases}
     \{ 0,1,\ldots, j\}, &\quad \text{\rm if }j\text{\rm~is an integer};\\
     \{1/2, 3/2, \ldots, j\}, &\quad \text{\rm if }j\text{\rm~is a half-integer}.
\end{cases}
}
\end{lemma}

\begin{proof} 

Let $\{\ket{j,m}_z\}_{m\in\mathcal{M}_j}$ denote a basis in which $J^{(j)}_z$ is diagonal, i.e., $J_z^{(j)}\ket{j,m}_z = m\ket{j,m}_z$, let $\Pi^{(j)}_\text{even}$ and $\Pi^{(j)}_\text{odd}$ be the projectors onto $W^{(j)}_\text{even}$ and $W^{(j)}_\text{odd}$, respectively, and let $P^{(j)} = \Pi^{(j)}_\text{even}-\Pi^{(j)}_\text{odd}$. Then $P^{(j)} = e^{-i\pi j} e^{i \pi J^{(j)}_x}$ and, from the fact that $e^{i \pi J_x^{(j)}}J_z^{(j)}e^{-i\pi J_x^{(j)}}= - J_z^{(j)}$, it follows that $P^{(j)} J_z^{(j)} = -J^{(j)}_z P^{(j)}$. Thus,
\eq{
J_z^{(j)}P^{(j)} \ket{j, m}_z &= -P^{(j)}J_z^{(j)} \ket{j,m}_z= - m P^{(j)} \ket{j,m}_z,
}
i.e., $P^{(j)}\ket{j,m}_z$ is an eigenvector of $J_z^{(j)}$ with eigenvalue $-m$.  Defining $\ket{\psi^{(j)}_m}= \ket{j,m}_z+P^{(j)}\ket{j,m}_z$ we have
\eq{
P^{(j)}\ket{\psi^{(j)}_m}  &= \ket{\psi^{(j)}_m}
}
and
\eq{
\left(J^{(j)}_z\right)^2\ket{\psi^{(j)}_m} &= m^2\ket{\psi^{(j)}_m}.
}
Noting that $\ket{\psi^{(j)}_m}$ and $\ket{\psi^{(j)}_{m'}}$ are orthogonal when $m\neq m'$, the set $\left\{\ket{\psi^{(j)}_m}\right\}_{m\in \mathcal{M}_j^+}$ therefore constitutes an orthogonal basis for $W^{(j)}_\text{even}$ (which has dimension $\left\lfloor j\right\rfloor +1$). With respect to this basis, the eigenvalues of $\left(J_z^{(j)}\right)^2$ are $m^2$ for all $m\in \mathcal{M}_j^+$.

\end{proof}


\end{document}